\documentclass[runningheads]{llncs}
\usepackage{graphicx}
\usepackage{amsmath}
\usepackage{amssymb}
\usepackage{algorithm}
\usepackage{xspace}
\usepackage{algpseudocode}
\usepackage{hyperref}
\usepackage{tikz}
\usepackage{pgfplots}
\pgfplotsset{compat=newest}
\usepackage{thmtools, thm-restate}

\usepackage[margin=1.0in]{geometry}
\renewcommand\normalsize{\fontsize{11}{13}\selectfont}
\AtEndEnvironment{proof}{\qed}

\newcommand{\smallv}{\mathcal S}
\newcommand{\largev}{\mathcal L}

\newcommand{\minbag}{\mathcal B}
\newcommand{\interval}{\mathcal I}

\newcommand{\alphagw}{\alpha_{GW}}
\newcommand{\critangle}{\theta_c}
\newcommand{\maxcut}{\text{mc}}

\newcommand{\prob}{\mathbb P}
\newcommand{\expect}{\mathbb E}

\newcommand{\npc}{\textsc{NP-Complete}\xspace}

\newcommand{\mcproblem}{\textsc{Maximum Cut}\xspace}

\begin{document}
\title{Approximating Maximum Cut on Interval Graphs and Split Graphs beyond Goemans-Williamson}
\titlerunning{Approximating Maximum Cut on Interval Graphs and Split Graphs beyond Goemans-Williamson}
\author{
Jungho Ahn\inst{1} \and
Ian DeHaan\inst{2} \and
Eun Jung Kim\inst{3} \and
Euiwoong Lee\inst{2}}
\authorrunning{J. Ahn, I. DeHaan, E. Kim, and E. Lee}
\institute{
Durham University, UK
\thanks{Supported by Leverhulme Trust Research Project Grant RPG-2024-182.
\email{jungho.ahn@durham.ac.uk}
}
\and
University of Michigan, USA 
\thanks{
Supported in part by NSF grant CCF-2236669.
\email{\{idehaan, euiwoong\}@umich.edu}
}
\and
KAIST and IBS, South Korea $|$ CNRS, France
\thanks{Supported by Institute for Basic Science (IBS-R029-C1) and National Research Foundation, Korea (RS-2025-00563533).
\email{eunjung.kim@kaist.ac.kr }}
}
\maketitle              %

\begin{abstract}
We present a polynomial-time $(\alphagw + \varepsilon)$-approximation algorithm for the \mcproblem problem on interval graphs and split graphs, where $\alphagw \approx 0.878$ is the approximation guarantee of the Goemans-Williamson algorithm and $\varepsilon > 10^{-34}$ is a fixed constant.
To attain this, we give an improved analysis of a slight modification of the Goemans-Williamson algorithm for graphs in which triangles can be packed into a constant fraction of their edges.
We then pair this analysis with structural results showing that both interval graphs and split graphs either have such a triangle packing or have maximum cut close to their number of edges.
We also show that, subject to the Small Set Expansion Hypothesis, there exists a constant $c > 0$ such that there is no polyomial-time $(1 - c)$-approximation for \mcproblem on split graphs.
\end{abstract}

\section{Introduction}

Given a graph $G = (V, E)$, the \mcproblem problem asks for a subset of vertices $S \subseteq V$ that maximizes the number of edges with exactly one endpoint contained in $S$. 
In this paper, we study the approximability of the \mcproblem problem on interval graphs and split graphs.

A graph $G = (V, E)$ is \emph{interval} if there exists a collection of intervals on the real line $\{ \interval_v\}_{v \in V}$ such that $uv \in E$ if and only if $\interval_u \cap \interval_v \neq \emptyset$.
See Figure~\ref{fig:interval} for an example.
Interval graphs are used in the field of biology, where they model natural phenomena such as DNA and food webs \cite{JUNGCK20151}. 
They have also been used in the study of register allocation, where vertices correspond to variables and intervals correspond to ``live ranges'' \cite{hendren1992register}.
Finally, they have numerous desirable theoretical properties and have arisen as a natural class of graphs to design algorithms for.
For example, it is shown in \cite{feder1998list} that a certain variant of the graph homomorphism problem is polynomial-time solvable if and only if the label graph is interval. 

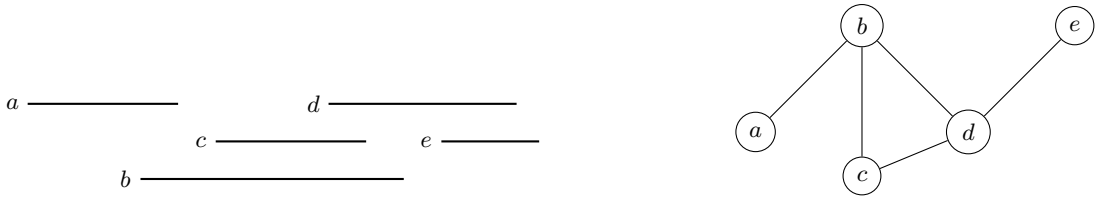
\begin{figure}[h] 
\centering

\begin{minipage}[b]{0.48\textwidth}
\centering
\begin{tikzpicture}

\draw[thick] (0, 2) -- (2, 2);   %
\draw[thick] (1.5, 1) -- (5, 1);       %
\draw[thick] (2.5, 1.5) -- (4.5, 1.5); %
\draw[thick] (4, 2) -- (6.5, 2);       %
\draw[thick] (5.5, 1.5) -- (6.8, 1.5); %

\node at (-0.2, 2) {$a$};
\node at (1.3, 1) {$b$};
\node at (2.3, 1.5) {$c$};
\node at (3.8, 2) {$d$};
\node at (5.3, 1.5) {$e$};

\end{tikzpicture}
\end{minipage}%
\hfill
\begin{minipage}[b]{0.48\textwidth}
\centering
\begin{tikzpicture}[every node/.style={circle, draw}, node distance=2cm]
\node (A) {$a$};
\node (B) [above right of=A] {$b$};
\node (C) [below of=B] {$c$};
\node (D) [below right of=B] {$d$};
\node (E) [above right of=D] {$e$};

\draw (A) -- (B);
\draw (B) -- (C);
\draw (B) -- (D);
\draw (C) -- (D);
\draw (D) -- (E);

\end{tikzpicture}
\end{minipage}

\caption{An interval graph. The left figure shows the interval representation. The right figure shows the resulting graph.} \label{fig:interval}
\end{figure}

A graph $G = (V, E)$ is \emph{split} if there exists a partition of $V = K \sqcup I$ such that $G[K]$, the graph induced by $K$, is a clique and $G[I]$, the graph induced by $I$, is independent. See Figure~\ref{fig:split} for an example.
Split graphs and interval graphs are both important subclasses of \emph{chordal graphs}, which themselves are a subclass of \emph{perfect graphs}.
In particular, split graphs are often the ``simplest'' subclass of perfect graphs in which problems are difficult to approximate.
Thus, considering split graphs is a natural first step when attempting to characterize a problem on chordal or perfect graphs.
In this paper, we show that, assuming the Small Set Expansion Hypothesis, there is some constant $c > 0$ such that there is no polynomial-time $(1-c)$-approximation for \mcproblem on split graphs.
To our knowledge, this is the first known hardness of approximation result for \mcproblem on perfect graphs.

\begin{figure}[h] 
\centering
\begin{tikzpicture}[every node/.style={circle, draw, inner sep=2pt}]

\node (C1) at (0,1) {$k_1$};
\node (C2) at (1,1) {$k_2$};
\node (C3) at (0,0) {$k_3$};
\node (C4) at (1,0) {$k_4$};

\node (I1) at (3,1.5) {$i_1$};
\node (I2) at (3,0.5) {$i_2$};
\node (I3) at (3,-0.5) {$i_3$};

\draw (C1) -- (C2);
\draw (C1) -- (C3);
\draw (C1) -- (C4);
\draw (C2) -- (C3);
\draw (C2) -- (C4);
\draw (C3) -- (C4);

\draw (C2) -- (I1);
\draw (C4) -- (I1);
\draw (C2) -- (I2);
\draw (C4) -- (I3);

\end{tikzpicture}
\caption{A split graph with $K = \{k_1, k_2, k_3, k_4\}$ and $I = \{i_1, i_2, i_3\}$.} \label{fig:split}
\end{figure}
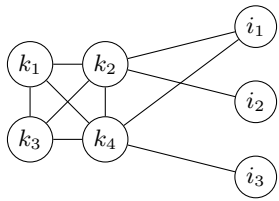

\mcproblem is one of the original 21 problems shown to be \npc by Karp \cite{Karp1972} and has long been a staple problem among algorithm researchers.
The seminal work of Goemans and Williamson shows that there is a polynomial-time $(\alphagw \approx 0.878)$-approximation algorithm for \mcproblem on general graphs \cite{goemans1995improved}.
The optimality of this result was open until 2007, when Khot et al. showed that, subject to the Unique Games Conjecture, there is no polynomial-time approximation algorithm for \mcproblem with an approximation ratio better than $\alphagw$ \cite{khot2007optimal}.

Remarkably, there is relatively little known about \mcproblem when the input graph is restricted to graphs from some structured class.
Even for subclass of perfect graphs, where problems such as \textsc{Independent Set} and \textsc{Chromatic number} admit polynomial-time algorithms based on semidefinite programming, \mcproblem, another flagship application of semidefinite programming, remains mostly unexplored.
In particular the extremely well-structured classes of interval graphs and split graphs, two important subclasses of perfect graphs, had no known approximation algorithm with a ratio better than $\alphagw$ prior to this work.
Our main results provide the first improved approximation for these classes of graphs since the work of Goemans and Williamson.

\begin{theorem} \label{thm:intervalapprox}
    There is a polynomial-time $(\alphagw + 10^{-34})$-approximation algorithm for \mcproblem on interval graphs.
\end{theorem}

\begin{theorem} \label{thm:splitapprox}
    There is a polynomial-time $(\alphagw + 10^{-16})$-approximation algorithm for \mcproblem on split graphs.
\end{theorem}

On the hardness side, it was shown by Bodlander and Jansen in 2000 that \mcproblem on split graphs is \npc \cite{bodlaender2000complexity}. 
\mcproblem on interval graphs has seen further attention lately - it was shown only recently that \mcproblem on interval graphs is \npc \cite{adhikary2023complexity}.
Further work refined this hardness for interval graphs which have at most $4$ interval lengths \cite{de2024maximum}, and later at most $2$ interval lengths \cite{barsukov2022maximum}.
Hardness for unit interval graphs - interval graphs with only $1$ interval length - remains open.
We remark that none of these hardness results imply any hardness of approximation.
Our final main result is that, subject to the Small Set Expansion Hypothesis of \cite{raghavendra2010graph}, \mcproblem on split graphs is hard to approximate to some constant factor.

\begin{restatable}{theorem}{splithard} \label{thm:splithard}
    There exists a constant $c > 0$ such that there is no polynomial-time $(1 - c)$-approximation algorithm for \mcproblem on split graphs, subject to the Small Set Expansion Hypothesis under randomized reductions.
\end{restatable}

\begin{table}
\centering
\def\arraystretch{1.5}

\begin{tabular}{|c|c|c|}
\hline
                & lower bound    & upper bound            \\ \hline
General graphs  & $\alphagw$ \cite{goemans1995improved}    & $\alphagw^\dagger$ \cite{khot2007optimal} \\ \hline
Degree $ \leq d$ graphs     & $\alphagw + \Omega(\frac{1}{d^2 \log d})$ \cite{hsieh2022approximating}     & $\alphagw + \mathcal O(\frac{1}{\sqrt d})^\dagger$  \cite{trevisan2001non}        \\  \hline
Interval graphs & $\mathbf{\boldsymbol\alphagw + 10^{-34}}$ & \npc \cite{adhikary2023complexity}        \\ \hline
Split graphs & $\mathbf{\boldsymbol\alphagw + 10^{-16}}$ & ${\bf 1 - c}^\star$     \\ \hline
Planar graphs   & 1 \cite{hadlock1975finding}    &          \\ \hline
Line graphs   & 1 \cite{guruswami1999maximum}    &          \\ 
\hline
\end{tabular}

\caption{Known results for \mcproblem. Our work appears in bold. Results marked with $^\dagger$ are subject to the Unique Games Conjecture. Results marked with $^\star$ are subject to the Small Set Expansion Hypothesis under randomized reductions.}
\label{table:results}

\end{table}

\subsection{Our Techniques} \label{subsection:techniques}

\subsubsection{Analysis of the Perturbed Goemans-Williamson Algorithm.}
Our starting point is the Goemans-Williamson algorithm \cite{goemans1995improved}.
This algorithm first solves the following semidefinite program (SDP)
\[ 
    {\bf maximize} \left\{\frac{1}{2}\sum_{uv \in E} (1 - x_u \cdot x_v) \mid x_v \in \mathbb S^{|V|-1} ~\forall~v \in V \right\}. 
\]

That is, the algorithm maps each vertex to a unit vector in a way that maximizes the sum of $1 - x_u \cdot x_v = 1 - \cos \theta_{uv}$, where $\theta_{uv}$ is the angle between $x_u$ and $x_v$.
Next, the algorithm samples a random Gaussian vector $r \sim \mathcal N(0, 1)^{|V|}$, creates the set $S = \{v \mid r \cdot x_v \geq 0\}$, and returns the cut defined by $S$.
This step is equivalent to sampling a random hyperplane and taking $S$ to be the vertices on one side of this hyperplane.
It is a straightforward calculation to see that the probability of an edge $e \in E$ being cut by $S$ is equal to $\theta_e/\pi$.
Thus, the approximation guarantee of the Goemans-Williamson algorithm is equal to
\[
     \alphagw := \min_{\theta \in [0, \pi]} \frac{2}{\pi} \frac{\theta}{1 - \cos \theta} \approx 0.878.
\]

Define $\critangle := \text{argmin}_{\theta \in [0, \pi]} \frac{\theta}{1- \cos \theta} \approx 134^{\circ}$ to be the ``critical angle'' at which this ratio is minimized.
We plot the performance guarantee over all angles in Figure~\ref{fig:ratio}.
Any edge $e \in E$ with $\theta_e \neq \critangle$ has approximation ratio strictly better than $\alphagw$.
Intuitively, if at least a constant fraction of edges have angle bounded away from $\theta_c$, we should expect the Goemans-Williamson algorithm to achieve a better approximation ratio.
Unfortunately, this is not exactly true.

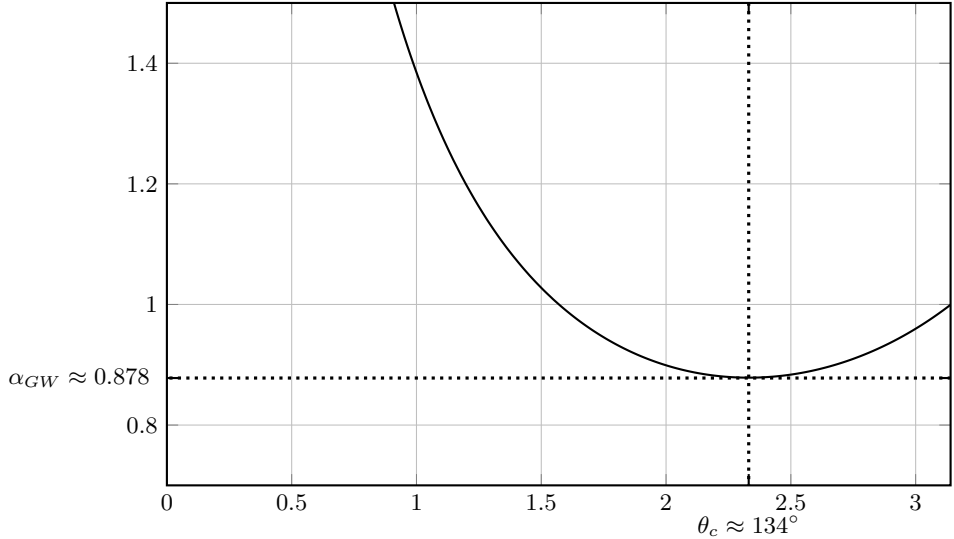
\begin{figure}[h!]
\centering
\begin{tikzpicture}
\begin{axis}[
    domain=0.01:3.14,        %
    samples=400,
    ymin=0.7, ymax=1.5,
    xmin=0, xmax=3.14,
    grid=major,
    thick,
    width=12cm,
    height=8cm,
    legend style={at={(0.5,-0.15)}, anchor=north, /tikz/every even column/.append style={column sep=1cm}},
    extra y ticks={0.878}, 
    extra y tick style={grid=none, tick label style={left=3pt}, tick style={black, thick}},
    extra y tick labels={$\alphagw \approx 0.878$},
    extra x ticks={2.33112},
    extra x tick style={grid=none, tick label style={below=8pt}, tick style={black, thick}},
    extra x tick labels={$\critangle \approx 134^{\circ}$}
]
  \addplot[
    black,
    thick
  ]
  {(2/pi) * (x / (1 - cos(deg(x))))};

  \addplot[
    black,
    dotted,
    very thick
  ]
  {0.878};

  \addplot[
    black,
    dotted,
    very thick
  ]
  coordinates {(2.33112,0.5) (2.33112,2)};
\end{axis}
\end{tikzpicture}
\caption{Plot of $\frac{2}{\pi}\frac{\theta}{1-\cos(\theta)}$ for $\theta\in(0,\pi)$.} \label{fig:ratio}
\end{figure}

Consider the graph $P_3$, the path on $3$ vertices.
Suppose the SDP is solved suboptimally, so that the first edge has angle $\critangle$ and the second edge has angle $0$.
Even though half of the edges have angle bounded away from $\critangle$, the expected size of the cut from rounding is still only $\alphagw$ times the optimal value of the SDP.
To handle situations like this, we introduce a ``perturbed'' version of the Goemans-Williamson algorithm.
In this perturbed algorithm, any vertex $v \in V$ with $|r \cdot x_v| < \eta$ for some small $\eta$ will instead be included in $S$ with probability $\frac{1}{2}$.
The motivation behind this perturbation is to consider what happens when every $x_v$ is ``moved'' by a small amount.
If $x_u$ and $x_v$ are close together, then they are likely to move further apart after this perturbation. Inversely, if $x_u$ and $x_v$ are far apart, they are likely to move closer together after this perturbation.

Indeed, in Lemma~\ref{lem:smallanglesbetter}, we show that the perturbed algorithm is more likely to cut any edge $e \in E$ with $\theta_e < \pi/2$, and less likely to cut any edge with $\theta_e > \pi/2$.
Moreover, the further away $\theta_e$ is from $\pi/2$, the more the results of the perturbed algorithm vary from the unperturbed algorithm.
Thus, if there are many edges with angle near $0$, but few edges with angle near $\pi$, the perturbed algorithm will achieve an approximation ratio above $\alphagw$, see Lemma~\ref{lem:manysmall}.
On the other hand, if there are many edges with angle near $\pi$, then the Goemans-Williamson algorithm itself will  achieve an approximation ratio above $\alphagw$, see Lemma~\ref{lem:manylarge}.
Thus, if we take the maximum result of the perturbed and unperturbed Goemans-Williamson algorithm, we can ignore the case of having many $0$ edges which contribute little to the optimal.
That is, as shown in Lemma~\ref{lem:betterrounding}, we can indeed assume that if at least a constant fraction of edges have angle bounded away from $\critangle$, the Goemans-Williamson algorithm will achieve an approximation ratio above $\alphagw$.

Consider a single triangle $T = \{uv, vw, wu\}$. A straightforward analysis shows that \[\theta_{uv} + \theta_{vw} + \theta_{wu} \leq 360^{\circ},\] regardless of the values of $x_u, x_v, x_w$.
Thus, there is some edge $e \in T$ with $\theta_e \leq 120^{\circ} < \theta_c$. 
That is, for every triangle in our graph, we can expect at least one of its edges to have an angle bounded away from $\critangle$.
This leads us into Section~\ref{section:tradeoff}, where we show that every interval graph and split graph either have a large edge-disjoint triangle packing, or have a cut with at least $0.9 |E|$ edges.

\subsubsection{Finding an Edge-Disjoint Triangle Packing.}

Suppose $G = (K \cup I, E)$ is a split graph.
If at least $0.1 |E|$ of its edges are contained in the clique $K$, then we can pack triangles almost perfectly into those edges.
Otherwise, we have that at least $0.9 |E|$ edges are crossing between $K$ and $I$, and so we have found a cut with $0.9 |E| \geq 0.9 \cdot \maxcut(G)$ edges, where $\maxcut(G)$ is the size of a maximum cut of $G$.
Thus, we find ourselves in a ``win-win'' situation, where $G$ either has an edge-disjoint triangle packing on a constant fraction of its edges, or $G$ is nearly bipartite, and we can find a large cut. 
This analysis leads directly to an improved approximation for \mcproblem on split graphs.

It turns out that the same ``win-win'' structural result also holds for interval graphs, which we prove in Theorem~\ref{thm:intervaltradeoff}.
To show this, we employ a marking scheme, where each vertex is classified as either ``small'' or ``large'' based on how it interacts with the rest of the graph.
We first show in Lemma~\ref{lem:fewlarge} that the number of edges between large vertices is low compared to $|E|$.
The situation for small vertices is more complicated.
We show in Lemmas~\ref{lem:fewsmall} and~\ref{lem:buildtriangle} that one of the following conditions must always hold:
\begin{enumerate}
    \item $G$ has a bridge; or \label{cond1}
    \item there is a clique $C$ in $G$ such that the sum of degrees $\sum_{v \in C} d_v$ is not much more than $|C|^2$; or \label{cond2}
    \item the number of edges between small vertices is low compared to $|E|$. \label{cond3}
\end{enumerate}
If Condition~\ref{cond1} holds, we can essentially delete the bridge; because adding a bridge to a bipartite graph does not create an odd cycle, we can always add back in the bridge after finding a cut in the rest of the graph.
If Condition~\ref{cond2} holds and $|C| \geq 3$, then we can pack edge-disjoint triangles into $C$ and delete $C$ from the graph.
Because the sum of degrees is not much more than $|C|^2$, we have packed triangles into at least a constant fraction of the deleted edges.
The case of $|C| \leq 2$ is more tricky, as one cannot pack triangles into one or two vertices,  and must be handled separately.
If Condition~\ref{cond3} holds, then most of the edges in $G$ go between small and large vertices.
If we have already packed triangles into a constant fraction of edges via Condition~\ref{cond2}, then we are done.
Otherwise, ``most'' edges are remaining in the graph, and we have found a cut on ``most'' of these remaining edges.
For the right definition of ``most,'' this is a large cut of at least $0.9 |E|$ edges.

\subsubsection{Hardness of Approximating Maximum Cut on Split Graphs.}

To show that \mcproblem is hard to approximate on split graphs, we start from the following hardness result that follows using standard techniques from \cite{raghavendra2012reductions}.
Assuming the Small Set Expansion Hypothesis holds under randomized reductions, for any sufficiently small $\varepsilon > 0$, there is no polynomial time algorithm that, given a graph $G = (V, E)$, can distinguish between the following two cases.
\begin{enumerate}
    \item There exists a cut $S \subseteq V$ with $|S| \approx 0.5 |V|$ such that $|\delta(S)| \leq \mathcal O( \varepsilon |E|)$. \label{case1}
    \item For all cuts $S \subseteq V$ with $|S| \leq 0.5 |V|$, either $|S| \leq 0.2 |V|$ or $|\delta(S)| \geq \Omega( \sqrt \varepsilon |E|)$. \label{case2}
\end{enumerate}
That is, in Case~\ref{case1}, there is a cut with nearly half the vertices that cuts very few edges.
In Case~\ref{case2}, all cuts with at least a constant fraction of vertices must cut many more edges.
We transform $G$ into a split graph $G' = (K \cup I, E')$ by turning $V$ into a clique and setting $K := V$, setting $I := \{v_e \mid e \in E\}$, and adding an edge from $v_e$ to each endpoint of $e$.
Any cut of $G$ can have at most $0.25 |V|^2$ edges from those edges contained in $K$, and at most $2 |E|$ edges from those going between $I$ and $K$.
In Case~\ref{case1}, there exists a cut that does cut nearly $0.25 |V|^2  + 2 |E|$ edges.
In Case~\ref{case2}, any cut of $G$ either cuts at most $0.2 |V|^2$ edges from those edges contained in $K$, or at most $(2 - \Omega(\sqrt \varepsilon))|E|$ from those edges crossing between $I$ and $K$.
After ensuring that $|E| \approx |V|^2$, this shows that \mcproblem is hard to approximate on split graphs.

\subsection{Preliminaries}

\subsubsection{Graphs.}
We consider only finite graphs in this paper. 
Apart from Section~\ref{section:hardness}, all considered graphs will also be simple and unweighted.
When the subject graph $G$ is clear from context, we will use $V := V(G)$ to refer to the vertex set of $G$, $E := E(G)$ to refer to the edge set of $G$, and $n := |V|$ to refer to the number of vertices in $G$.
Let $S \subseteq V$ be a subset of vertices of $G$.
We define $G[S] := (S, \{uv \in E \mid u, v \in S\})$ as the subgraph of $G$ induced by $S$ and $E[S] := E(G[S])$ as the set of edges in this subgraph.
We define $\delta_G(S) := \{uv \in E \mid u \in S, v \not \in S\}$ to be the cut induced by $S$ and $\maxcut(G) := \max_{S' \subseteq V}\{|\delta_G(S')|\}$ to be the maximum cut size of $G$.
We will often omit $G$ and write only $\delta(S)$ when the graph $G$ is clear from context.
Let $v \in V$ be any vertex of $G$.
We define $\delta(v) := \delta(\{v\})$ as the set of edges adjacent to $v$, $N(v) := \{u \in V \mid uv \in E\}$ as the set of neighbors of $v$, and $d_v := |N(v)|$ as the degree of $v$.
We define a triangle of $G$ as a set of edges $\{uv, vw, wu\} \subseteq E$ forming a triangle.
We say that $G$ has a \emph{triangle packing} of size $t$ if there exist $t$ edge-disjoint triangles $T_1, T_2, \ldots, T_t \subseteq E$.

\subsubsection{Gaussians.}
We define $\mathcal N(0, 1)$ as the Gaussian distribution with mean $0$ and variance $1$.
Moreover, we define $\mathcal N(0, 1)^n$ as the $n$-dimensional Guassian distribution, where a sampled vector $v \sim \mathcal N(0, 1)^n$ has $v_i \sim \mathcal N(0, 1)$ for each $i \in [n]$ and $v_i$ is independent of $v_j$ for $i \neq j$.
We make use of the fact that sampling a vector in this way is equivalent to sampling a random direction; that is, after normalizing, $v$ becomes a uniformly random unit vector in $n$-dimensional space.
In particular, this means that $\mathcal N(0, 1)^n$ is symmetric up to rotations, which we will exploit frequently.
We will also make use of the following lemma, which says that $\mathcal N(0, 1)$ is roughly equivalent to a uniform distribution close to $0$.
\begin{lemma} \label{lem:gaussian}
    For a randomly sampled $r \sim \mathcal N(0, 1)$, we have that
    \[
        \frac{x}{2} \leq \prob[|r| \leq x] \leq x
    \]
    for all $x \in [0, 1]$.
\end{lemma}
The proof of Lemma~\ref{lem:gaussian} follows from direct calculation and is not instructive, so we omit it.

\section{Triangle Packing and Maximum Cut Tradeoff} \label{section:tradeoff}

This section is devoted to proving the following structural result.%
\begin{theorem} \label{thm:intervaltradeoff}
    If $G = (V, E)$ is an interval graph, then $G$ either has a triangle packing of size $10^{-8} |E|$ or has a cut of size at least $0.9 |E|$ that can be found in polynomial time.
\end{theorem}

\subsection{Warmup: Tradeoff for Split Graphs}

As a warmup, we prove the following structural result that is essentially equivalent to Theorem~\ref{thm:intervaltradeoff}, except it is for split graphs instead of interval graphs.
\begin{theorem} \label{thm:splittradeoff}
    If $G = (V, E)$ is a split graph, then $G$ either has a triangle packing of size $0.01 |E|$ or has a cut of size at least $0.9 |E|$ that can be found in polynomial time.
\end{theorem}

Before we prove this, we need the following helpful lemma.
\begin{lemma} \label{lem:completetrianglepacking}
    The complete graph $K_n$ on $n \geq 3$ vertices has an edge-disjoint triangle packing of size $\frac{|E[K_n]|}{10} = \frac{n (n-1)}{20}$.
\end{lemma}

Edge-disjoint triangle packings in complete graphs have been studied in the literature before, with \cite{feder2012packing} giving an optimal bound.
Lemma~\ref{lem:completetrianglepacking} is not very close to the optimal bound, but it is sufficient for our purposes and substantially easier to prove.

\begin{proof}
    Label the vertices of $K_n$ with numbers $0, 1, 2, \ldots, n-1$.
    Label each triangle $\{uv, vw, wu\}$ with $(u + v + w) \text{ modulo } n$.
    Fix any edge $uv \in E(K_n)$. 
    Then we have that, for each possible triangle label, there is only one triangle involving $uv$ with that label.
    So we may take all the triangles of any specific label and find an edge-disjoint triangle packing.
    By the pigeon hole principle, some label has at least $\binom{n}{3} / n$ triangles. 
    Therefore, $K_n$ has a triangle packing of size $\binom{n}{3} / n \geq \frac{|E[K_n]|}{10}$, as wanted.
\end{proof}

With Lemma~\ref{lem:completetrianglepacking} in hand, we can now prove Theorem~\ref{thm:splittradeoff}.

\begin{proof}[Theorem~\ref{thm:splittradeoff}]
    Write $V = K \sqcup I$ where $K$ and $I$ are the clique and independent portions of $G$, respectively.
    Then we can partition $E = \delta(K) \sqcup E[K]$.
    If $|\delta(K)| \geq 0.9|E|$, then we are done, as we have constructed a cut of size at least $0.9 |E|$ in polynomial time.
    Otherwise, we have that $|E[K]| \geq 0.1 |E|$. 
    Now, recalling that $G[K]$ is the complete graph on $|K|$ vertices, we use Lemma~\ref{lem:completetrianglepacking} to find that $G[K]$ (and thus $G$) has an edge-disjoint triangle packing of size at least $\frac{|E[K]|}{10} \geq 0.01 |E|$.
\end{proof}

\subsection{Finding a Cut}

The proof of Theorem~\ref{thm:intervaltradeoff} maintains a similar flavor as that of Theorem~\ref{thm:splittradeoff}.
While the proof of Theorem~\ref{thm:splittradeoff} either finds a large cut or one large clique to pack triangles into, we need to repeatedly apply Lemma~\ref{lem:completetrianglepacking} during the proof of Theorem~\ref{thm:intervaltradeoff}, as there may be no single large-enough clique.
That is, at each stage, we either identify a clique that we may pack triangles into and remove while being careful to bound the number of non-clique edges we remove, or we conclude there is a large cut. 

Given an interval graph $G = (V, E)$, we proceed by partitioning the vertices into ``small'' and ``large'' vertices $V = \smallv \sqcup \largev$.
At each step, we will either certify that the implied cut is large $|\delta(\smallv)| \geq 0.9 |E|$, or find a way to expand a triangle packing on edges in $E[\smallv]$.

For any $t \in \mathbb{R}$, define $B_t := \{v \in V \mid t \in \interval_v\}$ to be the ``bag'' of vertices whose intervals occupy position $t$.
For a vertex $v \in V$, define $\minbag_v := \text{min}_{t \in \interval_v} |B_t|$ to be the size of the smallest bag in $v$'s interval.
Let $T$ be a constant we will decide the value of later.
We say that $v$ is ``small'' if $d_v \leq T \cdot \minbag_v$ and ``large'' otherwise.
In other words, if at least a constant fraction of $v$'s edges ``come from'' $\minbag_v$, then $v$ is small.
We define $\smallv$ as the set of small vertices and $\largev$ as the set of large vertices.

\begin{lemma} \label{lem:fewlarge}
    $|E[\largev]| \leq 8 T^{-1} \cdot |E|$.
\end{lemma}

\begin{proof}
    Define an ordering $\prec$ on $\largev$ by $u \prec v$ if $\minbag_u < \minbag_v$. 
    If $\minbag_u = \minbag_v$ we say $u \prec v$ if the leftmost point of $\interval_u$ is to the left of the leftmost point of $\interval_v$.
    For $u, v \in V$ with $\minbag_u = \minbag_v$ and equal leftmost point, we break ties arbitrarily to extend $\prec$ into a total ordering.
    
    Fix any $v \in \largev$. We will bound the number of edges $uv \in E[\largev]$ with $u \prec v$ as a function of $d_v$.
    Consider any $u \in \largev \cap N(v)$ such that $\interval_u$ does not intersect either the leftmost point or rightmost point of $\interval_v$.
    Then we must have that $\interval_u \subseteq \interval_v$ and thus $\minbag_u \geq \minbag_v$.
    Combined with the fact that the leftmost bag of $\interval_u$ is to the right of the leftmost bag of $\interval_v$ this implies that $v \prec u$.
    Thus, we need only consider $u \in \largev \cap N(v)$ such that $\interval_u$ intersects either the leftmost or rightmost bag of $\interval_v$.

    Let $P_v$ be the set of $u \in \largev \cap N(v)$ such that $\interval_u$ contains the leftmost point of $\interval_v$.
    Let $L_v \subseteq P_v$ be the first $\min\{|P_v|, \minbag_v\}$ of these neighbors sorted by increasing leftmost point.
    Similarly, let $R_v \subseteq P_v$ be the last $\min\{|P_v|, \minbag_v\}$ of these neighbors sorted by increasing rightmost point. 
    Now notice that all $u \in P_v \setminus (L_v \cup R_v)$ have $\minbag_u > \minbag_v$ and thus $v \prec u$.
    This is because every point of $\interval_u$ either intersects all of $L_v$ or all of $R_v$, which have size at least $\minbag_v$ assuming $P_v \setminus (L_v \cup R_v)$ is non-empty.
    See Figure~\ref{fig:overlapinterval} for an illustration.
    Thus, the number of vertices $u \in P_v$ such that $u \prec v$ is at most $|L_v \cup R_v| \leq 2 \minbag_v$.
    We can similarly bound the number of vertices $u \in \largev \cap N(v)$ such that $\interval_u$ contains the rightmost point of $\interval_v$ and $u \prec v$ by $2 \minbag_v$.
    
    Putting these bounds together with the fact that $v \in \largev$, we can bound the total number of edges $uv \in E[\largev]$ with $u \prec v$ by $4 \minbag_v \leq 4 T^{-1} \cdot d_v$.
    Thus, we find that $|E[\largev]| \leq \sum_{u \in \largev} 4 T^{-1} \cdot d_u \leq 8T^{-1} \cdot |E|$.

\usetikzlibrary{decorations.pathreplacing}
\begin{figure}[htbp]
  \centering
\begin{tikzpicture}[scale=1.5,>=stealth, every node/.style={font=\scriptsize}]
  \foreach \y/\llen/\rlen in {
    1.0/1.8/0.3,   %
    1.2/2.3/0.2,   %
    1.4/3.0/0.4    %
  }{
    \draw[thick] (-\llen,\y) -- (\rlen,\y);
  }
  \draw[decorate, decoration={brace, amplitude=3pt}]
    (-3.1,0.95) -- (-3.1,1.45) node[midway,left=4pt] {$L_v$};

  \draw[thick] (0.6,1.25) -- (6,1.25);
  \draw[thick] (0.6,1) -- (6,1);

  \foreach \y/\llen/\rlen in {
    0.6/0.3/1.6,   %
    0.4/0.2/2.5,   %
    0.2/0.4/3.2    %
  }{
    \draw[thick] (-\llen,\y) -- (\rlen,\y);
  }
  \draw[decorate, decoration={brace, amplitude=3pt}]
    (-0.5,0.15) -- (-0.5,0.65) node[midway,left=4pt] {$R_v$};

  \draw[thick] (0,0) -- (6,0) node[midway, below=1pt] {$\interval_v$};

  \draw[thick] (-1,-0.5) -- (1,-0.5) node[midway, below=1pt] {$\interval_u$};
\end{tikzpicture}

  \caption{Representation of $\interval_u, \interval_v, L_v, R_v$ with $\minbag_v = 3$. Every point of $\interval_u$ is contained in either all of $L_v$ or all of $R_v$.}
  \label{fig:overlapinterval}

\end{figure}
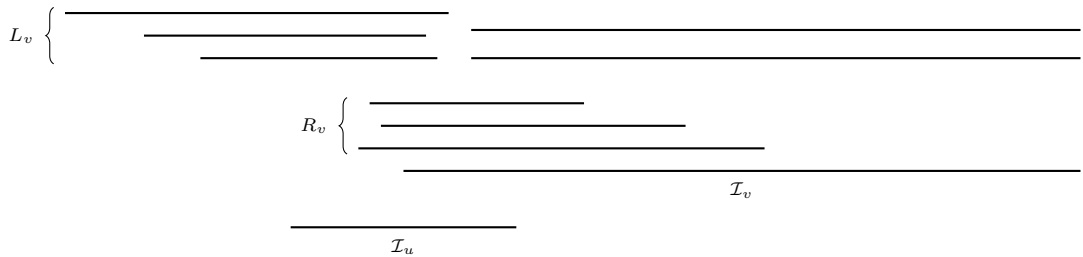

\end{proof}

Unlike with large vertices, we cannot unconditionally bound $|E[\smallv]|$.
We instead introduce a condition that, if unsatisfied, will allow us to make progress towards a triangle packing.

\begin{lemma} \label{lem:fewsmall}
    If, for some $\varepsilon > 0$, all $t \in \mathbb R$ have $|B_t \cap \smallv| \leq \max\{1, \varepsilon \cdot |B_t|\}$, then $|E[\smallv]| \leq 4 \varepsilon \cdot |E|$.
\end{lemma}

\begin{proof}
    Fix any non-isolated $v \in \smallv$.
    Let $t$ denote the leftmost point of $\interval_v$.
    Let $P_v \subseteq \smallv \cap N(v)$ denote the set of small neighbors of $v$ whose intervals contain $t$.
    Note that $\sum_{u \in \smallv} |P_u| \geq |E[\smallv]|$.
    By the definition of $B_t$, we can bound $|P_v| = |B_t \cap \smallv| - 1 \leq \varepsilon \cdot |B_t|$.
    Additionally, we have that $d_v \geq |B_t| - 1 \geq |B_t|/2$ and thus $|P_v| \leq 2 \varepsilon \cdot d_v$.
    Now we can bound $|E[\smallv]|$ by iterating over all small vertices $u \in \smallv$ \[
    |E[\smallv]| \leq \sum_{u \in \smallv} |P_u| \leq \sum_{u \in \smallv} 2 \varepsilon \cdot d_u \leq 4\varepsilon \cdot |E|.
    \]
\end{proof}

\subsection{Building a Triangle Packing}

\begin{lemma} \label{lem:buildtriangle}
    If, for some $\varepsilon > 0$ and $t \in \mathbb{R}$, $|B_t \cap \smallv| \geq \max\{2, \varepsilon \cdot |B_t|\}$, then either
    \begin{enumerate}
        \item  we can pack at least $\frac{\varepsilon}{30T} \cdot \sum_{u \in B_t \cap \smallv} d_u$ edge-disjoint triangles into $\bigcup_{u \in B_t \cap \smallv} \delta(u)$ or \label{condition1} %
        \item $G$ has a bridge. \label{condition2}
    \end{enumerate}
\end{lemma}

\begin{proof}
    Suppose $|B_t \cap \smallv| \geq 3$. 
    Then by Lemma~\ref{lem:completetrianglepacking}, we can pack at least $|B_t \cap \smallv|^2/30$ edge-disjoint triangles into $B_t \cap \smallv$.
    By the definition of $\smallv$, we have that 
    \[
        \sum_{u \in B_t \cap \smallv} d_u \leq \sum_{u \in B_t \cap \smallv} T \cdot \minbag_u \leq |B_t \cap \smallv| \cdot T \cdot |B_t| \leq |B_t \cap S|^2 \cdot T \cdot \varepsilon^{-1}.
    \]
    This fulfills Condition~\ref{condition1}.
    
    Now suppose $|B_t \cap \smallv| = 2$. 
    Label $B_t \cap \smallv = \{v_1, v_2\}$. 
    If $N(v_1) \cap N(v_2) = \emptyset$, then $v_1v_2$ is a bridge of $G$ and thus Condition~\ref{condition2} is fulfilled.
    Otherwise, let $u \in N(v_1) \cap N(v_2)$.
    We can pack a single triangle into the edges $v_1v_2, uv_1, uv_2$.
    Further, we have that $\varepsilon \cdot |B_t| \leq 2$, and so $d_{v_1} + d_{v_2} \leq \frac{4T}{\varepsilon}$.
    This fulfills Condition~\ref{condition1}.
\end{proof}

This leads us to Algorithm~\ref{alg:intervalmaxcut}.
We will iteratively find a clique to pack triangles into, delete the clique, and continue.
If we reach a point where we have deleted at least $0.01 |E|$ edges, then we can finish as we have packed triangles into a constant fraction of the edges.
Otherwise, if we run out of cliques to pack triangles into, we certify that we have an almost-complete cut on the remaining graph, which still contains most of the original edges of $G$.
Finally, whenever we identify a bridge, we can simply delete it from the graph and add it to our final cut, as bridges can always be added to any cut.

\begin{algorithm}[H]
\caption{$\textsc{IntervalMaxCut}(G = (V, E), T, \varepsilon)$}
\label{alg:intervalmaxcut}
\begin{algorithmic}
\State $G_0 \gets G$, $\mathcal T \gets \varnothing$, $A \gets \varnothing$, $i \gets 0$
\While{$|\mathcal T| \leq 0.01 |E|$}
    \State label $V(G_i) = \smallv \cup \largev$ as defined
    \If{$G_i$ has a bridge $e$}
        \State $A \gets A \cup \{e\}$
        \State $G_{i+1} \gets (V(G_i), E(G_i) \setminus \{e\})$
    \ElsIf{some $t \in \mathbb R$ has $|B_t \cap \smallv| \geq \max\{2, \varepsilon \cdot |B_t|\}$}
        \State $\mathcal T \gets \mathcal T \cup \bigcup_{u \in B_t \cap \smallv} \delta(u)$
        \State $G_{i+1} \gets G_i[V(G_i) \setminus (B_t \cap \smallv)]$
    \Else
        \State return cut $\delta(\smallv) \cup A$
    \EndIf
    \State $i \gets i+1$
\EndWhile
\State return $\textsc{PerturbGW}(G)$
\end{algorithmic}
\end{algorithm}
Note that Algorithm~\ref{alg:intervalmaxcut} runs in polynomial time.
There are at most $|E|$ iterations of the while loop, as each iteration either returns or removes at least one edge.

\begin{lemma} \label{lem:bigcut}
    If $\textsc{IntervalMaxCut}(G, T, \varepsilon)$ returns the set $\delta(\smallv) \cup A$ from within the while loop, then $|\delta(\smallv) \cup A| \geq 0.99 \cdot (1 - 4 \varepsilon - 8T^{-1}) \cdot |E|$ and $\delta(\smallv) \cup A$ is a valid cut.
\end{lemma}

\begin{proof}
    Suppose that Algorithm~\ref{alg:intervalmaxcut} returns on the $i$th iteration of the while loop.
    By Lemmas~\ref{lem:fewlarge} and~\ref{lem:fewsmall}, we have that $|\delta(\smallv)| \geq (1 - 4\varepsilon - 8T^{-1}) \cdot |E_i|$ edges.
    Note that $E = E_i \cup A \cup \mathcal T$ and due to the while loop condition, we have that $|\mathcal T| \leq 0.01 |E|$.
    Thus, $|\delta(\smallv) \cup A| \geq 0.99 \cdot (1 - 4 \varepsilon - 8T^{-1}) \cdot |E|$.

    To see that $\delta(\smallv) \cup A$ is a valid cut, note that adding a bridge to a bipartite graph cannot introduce an odd cycle.
    Thus, we can iteratively add each edge of $A$ to $\delta(\smallv)$ without invalidating our cut.
\end{proof}

\begin{lemma} \label{lem:bigtriangle}
    If $\textsc{IntervalMaxCut}(G)$ exits the while loop, then $G$ has an edge-disjoint triangle packing of size at least $\frac{\varepsilon}{3000T} |E|$.
\end{lemma}

\begin{proof}
    At the conclusion of the while loop, we have that $|\mathcal T| \geq 0.01 \cdot |E|$.
    Due to Lemma~\ref{lem:buildtriangle}, each time we expand $\mathcal T$ by $x$ edges, we can pack an additional $\frac{\varepsilon}{30T} x$ edge-disjoint triangles into $\mathcal T$.
    By iterating this process, there exists an edge-disjoint triangle packing of size at least $\frac{\varepsilon}{30T} |\mathcal T| \geq \frac{\varepsilon}{3000T} |E|$ in $G$.
\end{proof}

Now set $\varepsilon = 0.01$ and $T = 200$.
Lemmas~\ref{lem:bigcut} and~\ref{lem:bigtriangle} imply that either $G$ has a cut of size at least $0.99 \cdot 0.92 |E| > 0.9 |E|$ or an edge-disjoint triangle packing of size at least $\frac{1}{6} \cdot 10^{-7} |E| > 10^{-8} |E|$, as wanted.

\subsection{No Tradeoff for Chordal Graphs}

In this subsection, we will show that there is no equivalent tradeoff for chordal graphs as there are for interval graphs and split graphs.
Thus, improving upon $\alphagw$ for chordal graphs will likely require new algorithmic insights.

\begin{theorem}
    For all $c > 0$, there exists a chordal graph $G = (V, E)$ such that $\maxcut(G) < \alphagw |E|$ and $G$ has no edge-disjoint triangle packing of size $c |E|$.
\end{theorem}

\begin{proof}
    Suppose we are given a chordal graph $G' = (V', E')$ with $|V'| < c  |E'|$ %
    and no triangle packing of size $c |E'|$.
    We will show later how to obtain such a $G'$.
    Construct $G = (V, E)$ by setting $V := V' \cup \{w\} \cup \{x_e, y_e \mid e \in E'\}$ and $E := E' \cup \{x_e u, y_e v \mid e = uv \in E'\} \cup \{w v \mid v \in V \setminus \{w\}\}$.
    That is, for each edge $e$, we attach one fresh vertex to each endpoint of the edge, and create a ``universal'' vertex $w$ connected to all vertices in the graph.

    We first note that $G$ is chordal, as it is obtained from a chordal graph $G'$ by iteratively adding simplicial vertices.
    By inspection, any triangle $T \not \subseteq E'$ must contain at least one edge from $\{w v \mid v \in V'\}$.
    Thus, the maximum number of edge-disjoint triangles in $G$ is at most $|V'| + c  |E'| < 2c  |E'| < c  |E|$. 
    For each edge $e = uv \in E'$, $G$ contains a 5-cycle $w x_e, x_e u, uv, vy_e, y_e w$, and so $G$ contains $|E'|$ edge-disjoint 5-cycles. Thus, we have that $\maxcut(G) \leq |E| - |E'|$. 
    Note that $|E| = 5 |E'| + |V'| \leq 6 |E'|$, so $\maxcut(G) \leq \frac{5}{6} |E| < \alphagw |E|$.

    It remains to show that a chordal graph $G'$ with the desired properties exists.
    We give the following interval graph construction for $G'$. Select $k$ large enough. 
    Create $2^k - 1$ vertices in a ``segment-tree'' pattern as follows.
    In the first layer, create one vertex with interval $(0, 1)$.
    In the second layer, create two vertices with intervals $(0, 0.5)$ and $(0.5, 1)$ respectively.
    In the third layer, create four vertices with intervals $(0, 0.25), (0.25, 0.5), (0.5, 0.75)$ and $(0.75, 1)$ respectively.
    Then iterate this process for $k$ total layers, see Figure~\ref{fig:layered-intervals} for an illustration.

    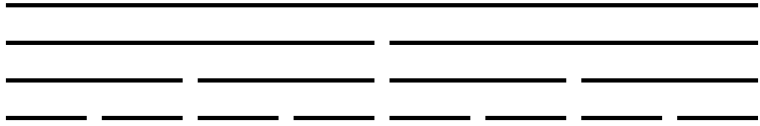
\begin{figure}[ht]
      \centering
      \begin{tikzpicture}[scale=10]
        \def\k{4}                %
        \def\spacing{0.05}       %
        \def\padding{0.00}       %
        \def\fixedgap{0.02}      %
    
        \foreach \i in {0,...,3} {
          \pgfmathtruncatemacro{\numintervals}{int(2^\i)}
          \pgfmathsetmacro{\y}{-\i * \spacing}
          
          \pgfmathsetmacro{\totalgap}{(\numintervals - 1) * \fixedgap}
          \pgfmathsetmacro{\layerwidth}{1 - 2 * \padding}
          \pgfmathsetmacro{\availablewidth}{\layerwidth - \totalgap}
          \pgfmathsetmacro{\intervalwidth}{\availablewidth / \numintervals}
    
          \foreach \j in {0,...,31} {
            \ifnum\j<\numintervals
              \pgfmathsetmacro{\xstart}{\padding + \j * (\intervalwidth + \fixedgap)}
              \pgfmathsetmacro{\xend}{\xstart + \intervalwidth}
              \draw[black, ultra thick] (\xstart,\y) -- (\xend,\y);
            \fi
          }
        }
      \end{tikzpicture}
      \caption{The interval representation of $G'$ for $k = 4$.}
      \label{fig:layered-intervals}
    \end{figure}

    We have that $|V'| = 2^k - 1$ and $|E'| \geq 2^{k-1} \cdot (k-1) > |V'| \cdot \frac{\log_2 |V'|}{2}$ by counting only the edges adjacent to the bottom layer. 
    Thus, for $k$ sufficiently large, we have that $|V'| < c |E'|$.

    To show that $G'$ has no edge-disjoint triangle packing of size $c  |E'|$, it is sufficient to show that $\maxcut(G') > (1 - \frac{c}{3}) \cdot |E'|$, as each triangle results in at least one un-cuttable edge.
    Consider the cut of $G'$ obtained by taking the bottom $t$ layers as one side of the cut.
    The number of edges in this cut is $(2^{k} - 2^{k-t}) \cdot (k - t)$.
    Note also that $|E'| \leq 2^k \cdot k$. 
    Therefore, we have that $\maxcut(G') \geq \frac{(2^k - 2^{k-t}) \cdot (k-t)}{2^k \cdot k} |E'|$.
    As $k$ tends to infinity, the term $\frac{(2^k - 2^{k-t}) \cdot (k-t)}{2^k \cdot k}$ tends to $1 - 2^{-t}$.
    So, by setting $t$ to be a sufficiently large constant based on $c$ and letting $k$ be sufficiently large based on $t$, we have that $\maxcut(G') > (1 - \frac{c}{3}) \cdot |E'|$ and so $G'$ has no edge-disjoint triangle packing of size $c  |E'|$.
    Thus, $G'$ fulfills all the conditions we needed to produce $G$.
\end{proof}

\section{Analysis of the Perturbed Goemans-Williamson Algorithm}

Our second main contribution is an improved approximation for \mcproblem in graphs with large triangle packings.
Given a fixed triangle $\{uv, vw, wu\}$, it is impossible for all angles $\theta_{uv}, \theta_{vw}, \theta_{wu}$ to be equal or very close to the critical angle $\critangle$.
However, it is possible for $\theta_{uv} = \theta_{vw} = \critangle$ to achieve the critical angle and have $\theta_{wu} = 0$. 
In this case, despite the entire graph being a single triangle, the Goemans-Williamson rounding algorithm will not perform beyond its worst-case guarantee. 
This is because the contribution of the edge $wu$ to the objective function is $0$, and so rounding it ``better'' does not actually increase our expected value.

To grapple with this issue, we introduce the ``Perturbed Goemans-Williamson Algorithm.''
Intuitively speaking, this algorithm randomly ``perturbs'' each vector slightly. 
We will see that edges with near-zero angle stand to gain much more in this perturbation process than any other edges have to lose besides those with an angle of nearly $\pi$.
Thus, any SDP solution with many near-zero angle edges and few near-$\pi$ angle edges can be rounded with a guarantee better than $\alphagw$.

Before presenting the algorithm, we must first define the semidefinite program from which we will round a solution.
\begin{alignat}{3}
\text{\bf maximize:} & \quad & \frac{1}{2} \sum_{uv \in E} (1 - x_u \cdot x_v) \tag{{\bf SDP-GW}} \label{SDP} \\
\text{\bf subject to:} && x_v \in \mathbb S^{n-1} \quad && ~\forall~ v \in V \notag
\end{alignat}

\begin{algorithm}[H]
\caption{$\textsc{PerturbGW}(G = (V, E), \eta)$}
\label{alg:perturbed}
\begin{algorithmic}
\State Solve \ref{SDP} and obtain optimal solution $\{x_v^*\}_{v \in V}$.
\State Sample a random $n$-dimensional vector $r \sim \mathcal N(0,1)^n$.
\For{all $v \in V$}
\State $s_v \gets \text{sign}(r \cdot x_v^*)$.
\If{$|r \cdot x_v^*| \geq \eta$}
\State $s_v' \gets s_v$.
\Else
\State Uniformly sample $s_v' \sim \{-1, 1\}$.
\EndIf
\EndFor
\State $S \gets \{v \mid s_v = 1\}, S' \gets \{v \mid s_v' = 1\}$.
\State Return \text{argmax}$_{\{\delta(S), \delta(S')\}}(|\delta(S)|, |\delta(S')|)$.
\end{algorithmic}
\end{algorithm}

Note that Algorithm~\ref{alg:perturbed} runs in polynomial time, as solving SDPs and sampling from a Gaussian distribution can be made to run in polynomial time.
\begin{restatable}{theorem}{apxratio} \label{thm:apxratio}
    If $G$ has an edge-disjoint triangle packing of size at least $t |E|$, then 
    \[
        \left|\textsc{PerturbGW}(G, \eta := \frac{t^2}{10^4})\right| \geq (\alphagw + 10^{-10} t^3) \cdot \maxcut(G).
    \]
\end{restatable}

Theorem~\ref{thm:apxratio}, when combined with Theorem~\ref{thm:intervaltradeoff} and Theorem~\ref{thm:splittradeoff}, immediately shows that there is a polynomial-time $(\alphagw+10^{-34})$-approximation for \mcproblem on interval graphs and a polynomial-time $(\alphagw+10^{-16})$-approximation for \mcproblem on split graphs.

Note that $\delta(S)$ is the result of running the original Goemans-Williamson algorithm, so the result of Algorithm~\ref{alg:perturbed} is immediately at least $\expect[|\delta(S)|] \geq \alphagw \cdot \maxcut(G)$.

For an edge $e \in E$, let $C_e := \mathbb I[e \in \delta(S)]$ and $C_e' := \mathbb I[e \in \delta(S')]$ be the random variables indicating that $e$ is cut by $S$ and $S'$, respectively.
Define $\theta_e := \arccos (x_u^* \cdot x_v^*)$ as the angle between $u$ and $v$.
For $\theta \in [0, \pi]$, let
\[
    E_{\theta} := \{e \in E \mid |\theta_e - \theta| \leq \sqrt \eta\}
\]
be the set of edges with angle ``close to'' $\theta$, where $\eta$ is a parameter of Algorithm~\ref{alg:perturbed}.
We first deal with the case where there are many edges in $E_{\pi}$.
\begin{lemma} \label{lem:manylarge}
    If $\eta \leq 0.01$ and $|E_\pi| \geq \eta^{3/2} |E|$, then
    $\expect[|\delta(S)|] \geq (\alphagw + 10^{-2} \eta^{3/2}) \cdot \maxcut(G)$.
\end{lemma}

\begin{proof}
    Let $SDP^*$ equal the value of \ref{SDP} at optimal solution $\{x_v^*\}_{v \in V}$.
    Consider any $e = uv \in E$.
    Recall that the contribution of $e$ to $SDP^*$ is $\frac{1 - x_u^* \cdot x_v^*}{2} = \frac{1 - \cos \theta_e}{2}$.
    Also, by simple calculation, we have that $\prob[e \in \delta(S)] = \frac{\theta_e}{\pi}$. 
    We calculate
    \begin{align*}
        \expect[|\delta(S)|] &= \sum_{e \in E} \frac{\theta_e}{\pi} \\
        &= \sum_{e \in E} \frac{1 - \cos \theta_e}{2} \cdot \frac{2 \theta_e}{\pi (1 - \cos \theta_e)} \\
        &= \sum_{e \in E \setminus E_{\pi}} \frac{1 - \cos \theta_e}{2} \cdot \frac{2 \theta_e}{\pi (1 - \cos \theta_e)} + \sum_{e \in E_\pi} \frac{1 - \cos \theta_e}{2} \cdot \frac{2 \theta_e}{\pi (1 - \cos \theta_e)} \\
        &\geq \alphagw \sum_{e \in E \setminus E_\pi} \frac{1 - \cos \theta_e}{2} + (\alphagw + 10^{-2}) \sum_{e \in E_\pi} \frac{1 - \cos \theta_e}{2} \\
        &= \alphagw \cdot SDP^* + 10^{-2} \sum_{e \in E_\pi} \frac{1 - \cos \theta_e}{2} \\
        &\geq \alphagw \cdot SDP^* + 10^{-2}  \cdot \frac{|E_\pi|}{|E|} \cdot SDP^*.
    \end{align*}
    The first inequality is by calculating $\frac{2}{\pi} \frac{\theta_e}{1  - \cos \theta_e} \geq \alphagw + 10^{-2}$ for $\theta_e > \pi - \sqrt{0.01}$.
   The second inequality follows from the fact that $1 - \cos \theta_e$ is increasing on $[0, \pi]$.
    The lemma now follows from the fact that $SDP^* \geq \maxcut(G)$.
\end{proof}

Now we show that for edges not in $E_{\pi}$, $S'$ is not much worse than $S$.
For technical reasons, our lemma statement also excludes edges in $E_0$.
We will see later that for all edges with angle at most $\frac{\pi}{2}$, $S'$ is no worse than $S$.
Additionally, for edges with angle very close to $0$, $S'$ is substantially better than $S$.

\begin{lemma} \label{lem:smalldecrease}
    For all $e \not \in E_0 \cup E_{\pi}$, we have that $\prob[C_e'] \geq \prob[C_e] - 10 \eta^{3/2}$.
\end{lemma}

\begin{proof}
    For any vertex $v \in V$, let $R_v := \mathbb I[|r \cdot x_v^*| < \eta]$ be the random variable indicating that $s_v'$ is randomly selected. %
    Fix any $e = uv \not \in E_0 \cup E_{\pi}$.
    We first note that $\prob[C_e' \mid \neg R_u \wedge \neg R_v] = \prob[C_e \mid \neg R_u \wedge \neg R_v]$ and $\prob[C_e' \mid R_u \vee R_v] = \frac{1}{2}$.

    We now turn our attention to $\prob[C_e \mid R_v]$.
    Note that $x_u^*$ and $x_v^*$ lie on a single plane through the origin, so by symmetry of $\mathbb S^{n-1}$, we can assume without loss of generality that $x_v^* = (1, 0, \ldots, 0)$ and $x_u^* = (x, y, 0, \ldots, 0)$ for $y \geq 0$.
    Label the random variable $r = (r_1, r_2, \ldots, r_n)$.
    By symmetry of output, we can assume without loss of generality that $r_1 \geq 0$.
    Note that these assumptions imply that $R_v$ is independent of the value of $r_2$.

    Suppose that $|r_2| > \frac{\eta}{y}$.
    Then we have that $|r_2 \cdot y| > \eta > |r_1 \cdot x|$, and so $\text{sign}(r \cdot x_u^*) = \text{sign}(r_2)$.
    That is, the sign of $r \cdot x_u^*$ is entirely determined by the sign of $r_2$.
    This implies that, when $|r_2| > \frac{\eta}{y}$,  $C_e = \mathbb I[\text{sign}(r_2) < 0]$.
    Thus, by independence of $R_v$ and the value of $r_2$,
    \[
        \prob[C_e \mid R_v \wedge |r_2| > \frac{\eta}{y}] = \prob[\text{sign}(r_2) < 0 \mid R_v \wedge |r_2| > \frac{\eta}{y}] = \prob[\text{sign}(r_2) < 0] = \frac{1}{2}.
    \]
    We then apply Lemma~\ref{lem:gaussian} to bound
    \[
        \prob[C_e \mid R_v] \leq \frac{1}{2} + \prob[|r_2| \leq \frac{\eta}{y} \mid R_v] \leq \frac{1}{2} + \frac{\eta}{y}.
    \]

    Let $M := \max\{\prob[C_e \mid R_u], \prob[C_e \mid R_v]\} \leq \frac{1}{2} + \frac{\eta}{y}$. We wish to bound $\prob[C_e \mid R_u \vee R_v] \leq M + \frac{4\eta}{y}$. To this end, we manipulate probabilities
    \begin{align*}
        \prob[C_e \wedge (R_u \vee R_v)] &\leq \prob[C_e \wedge R_u] + \prob[C_e \wedge R_v] \\
        &= \prob[C_e \mid R_u] \prob[R_u] + \prob[C_e \mid R_v] \prob[R_v] \\
        &\leq M (\prob[R_u] + \prob[R_v]) \\
        &= M (\prob[R_u \vee R_v] + \prob[R_u \wedge R_v]).
    \end{align*}
    Now we find
    \[
        \prob[C_e \mid R_u \vee R_v] = \frac{\prob[C_e \wedge (R_u \vee R_v)]}{\prob[R_u \vee R_v]} 
        \leq M + \frac{\prob[R_u \wedge R_v]}{\prob[R_u \vee R_v]}.
    \]

    To bound the numerator of the error term, first consider $\prob[R_u \mid R_v]$. 
    Recall that $R_u = \mathbb I[|r_1 x + r_2y| < \eta]$ and $R_v = \mathbb I[|r_1| < \eta]$. 
    Thus, assuming $R_v$ holds, in order for $R_u$ to hold, we must have that $|r_2 y| < 2 \eta$.
    Using Lemma~\ref{lem:gaussian}, we can bound
    \[
        \prob[R_u \mid R_v] \leq \prob[|r_2y | \leq 2 \eta \mid R_v] = \prob[|r_2| \leq \frac{2\eta}{y}] \leq \frac{2\eta}{y}.
    \]
    and using Lemma~\ref{lem:gaussian} again,
    \[
        \prob[R_u \wedge R_v] = \prob[R_u \mid R_v] \prob[R_v] \leq \frac{2\eta^2}{y}.
    \]
    Applying this with yet another application of Lemma~\ref{lem:gaussian} gives
    \[
        \frac{\prob[R_u \wedge R_v]}{\prob[R_u \vee R_v]} \leq \frac{4\eta}{y}.
    \]

    This yields the desired result, that $\prob[C_e \mid R_u \vee R_v] \leq \frac{1}{2} + \frac{5\eta}{y}$.
    
    Putting it all together, we find that
    \begin{align*}
        \prob[C_e] &= \prob[C_e \mid \neg R_u \wedge \neg R_v] \cdot \prob[\neg R_u \wedge \neg R_v] + \prob[C_e \mid  R_u \vee R_v] \cdot \prob[R_u \vee R_v]\\
        &\leq \prob[C_e \mid \neg R_u \wedge \neg R_v] \cdot \prob[\neg R_u \wedge \neg R_v] + (\frac{1}{2} + \frac{5\eta}{y}) \cdot \prob[R_u \vee R_v].
    \end{align*}
    and using Lemma~\ref{lem:gaussian},
    \begin{align*}
        \prob[C_e'] &= \prob[C_e' \mid \neg R_u \wedge \neg R_v] \cdot \prob[\neg R_u \wedge \neg R_v] + \prob[C_e' \mid R_u \vee R_v] \cdot \prob[R_u \vee R_v] \\
        &= \prob[C_e \mid \neg R_u \wedge \neg R_v] \cdot \prob[\neg R_u \wedge \neg R_v] + \frac{1}{2} \cdot \prob[R_u \vee R_v] \\
        &\geq \prob[C_e] - \frac{5\eta}{y} \cdot \prob[R_u \vee R_v] \\
        &\geq \prob[C_e] - \frac{10\eta^2}{y}.
    \end{align*}
    Noting that $y = \sin \theta_e \geq \sqrt \eta$ because $e \not \in E_0 \cup E_\pi$ completes the proof of the lemma.
\end{proof}

Lemma~\ref{lem:smalldecrease} allows us to bound the perturbation loss on all angles sufficiently bounded away from $0$ and $\pi$.
However, we will not be able to guarantee that the angles we consider are sufficiently bounded away from $0$ to properly utilize Lemma~\ref{lem:smalldecrease}.
Thus, we strengthen Lemma~\ref{lem:smalldecrease} to show that perturbation does not cause \emph{any} loss on angles below $\frac{\pi}{2}$.

\begin{lemma} \label{lem:smallanglesbetter}
    For all $e \in E$ with $\theta_e \leq \frac{\pi}{2}$, we have that $\prob[C_e'] \geq \prob[C_e]$.
\end{lemma}

\begin{proof}
    Take any $e = uv \in E$ such that $\theta_e \leq \frac{\pi}{2}$.
    As in the proof of Lemma~\ref{lem:smalldecrease}, restrict to two dimensions, rotate, and reflect so we may assume $x_v^* = (1, 0, \ldots, 0), x_u^* = (x, y, 0, \ldots, 0)$ for $y \geq 0$, and $r_1 \geq 0$.
    Due to our assumption on $\theta_e$, we have that $x = \cos \theta_e \geq 0$.
    As earlier, define the event $R_w := \mathbb I[|r \cdot x_w^*| < \eta]$ for $w \in V$.
    As in the proof of Lemma~\ref{lem:smalldecrease}, our main task is to bound $\prob[C_e \mid R_v]$.
    However this time, we will show $\prob[C_e \mid R_v] \leq \frac{1}{2}$.

    The event $R_v$ is equivalent to $\mathbb I[|r_1| < \eta]$.
    Since $x, y \geq 0$, if $C_e$ happens, then we must have $r_2 \leq 0$, so
    \[
        \prob[C_e | R_v] = \prob[r_1 x + r_2y < 0 \mid r_1 \in [0, \eta)] \leq \prob[r_2 \leq 0] = \frac{1}{2}.
    \]

    We have that
    \[
        R_u = \mathbb I[|r_1 x + r_2 y| < \eta] = \prob[r_2 y \in (-\eta - r_1 x, \eta - r_1 x)].
    \]
    Recall that $x \geq 0$ and $r_1 \geq 0$, so $|-\eta - r_1x| > |\eta - r_1 x|$.
    Thus, the event $R_u$ contains a larger portion of the negative space than the positive space and
    \[
        \prob[r_2 y < 0 \mid \neg R_u] \leq \frac{1}{2}.
    \]
    Recalling that $y > 0$, and $r_1$ and $r_2$ are independent, we can now calculate
    \begin{align*}
        \prob[C_e \mid R_v \wedge \neg R_u] &= \prob[r_1x + r_2 y < 0 \mid r_1 \in [0, \eta) \wedge \neg R_u] \\
        &\leq \prob[r_2 < 0 \mid r_1 \in [0, \eta) \wedge \neg R_u] \\
        &\leq \prob[r_2 < 0 \mid \neg R_u] \\
        &\leq \frac{1}{2}.
    \end{align*}
    and by symmetry, $\prob[C_e \mid R_v \wedge \neg R_u] \leq \frac{1}{2}$.
    A similar computation as that in the proof of Lemma~\ref{lem:smalldecrease} completes the proof.
\end{proof}

\begin{lemma} \label{lem:minvalue}
    For all $e \in E$, we have that $\prob[C_e'] \geq \frac{\eta}{4}$.
\end{lemma}

\begin{proof}
    Take any $e = uv \in E$ and calculate, using Lemma~\ref{lem:gaussian},
    \[
        \prob[C_e'] \geq \prob[C_e' \mid R_v] \cdot \prob[R_v] \geq \frac{1}{2} \cdot \frac{\eta}{2}.
    \]
\end{proof}

Now notice that if $\theta_e < \frac{\pi}{4} \eta$, then $\prob[C_e'] \geq \frac{\eta}{4} > \frac{\theta_e}{\pi} = \prob[C_e]$.
Thus, if a substantial fraction of $E$ has very small angle, then $\delta(S')$ will be noticeably larger than $\delta(S)$.
To this end, define
\[
    E_z := \{ e \in E : \theta_e \leq \frac{\pi}{8} \eta\}.
\]
Note that for each $e \in E_z$, we have $\prob[C_e] \leq \frac{\pi}{8} \eta \cdot \frac{1}{\pi} = \frac{\eta}{8} \leq 2 \prob[C_e']$.

\begin{lemma} \label{lem:manysmall}
    If $|E_z| \geq 96 \sqrt \eta |E|$ and $|E_\pi| < \eta^{3/2} |E|$, then 
    \[
        \expect[|\delta(S')|] \geq \mathbb \expect[|\delta(S)|] + \eta^{3/2} |E| \geq (\alphagw + \eta^{3/2}) \cdot \maxcut(G).
    \]
\end{lemma}

\begin{proof}
    Calculate, using Lemmas~\ref{lem:smalldecrease}, \ref{lem:smallanglesbetter}, and \ref{lem:minvalue}
    \begin{align*}
        \expect[|\delta(S')|] &= \sum_{e \in E_z} \prob[C_e'] + \sum_{e \in E_\pi} \prob[C_e'] + \sum_{e \in E \setminus (E_\pi \cup E_z)} \prob [C_e'] \\
        &\geq \sum_{e \in E_z} \frac{\eta}{4} + \sum_{e \in E \setminus (E_\pi \cup E_z)} \prob[C_e] - 10\eta^{3/2} \\
        &\geq \sum_{e \in E_z} (\prob[C_e] +\frac{\eta}{8}) + \sum_{e \in E \setminus (E_\pi \cup E_z)} \prob[C_e] - 10\eta^{3/2} \\
        &\geq \expect[|\delta(S)|]  + \frac{\eta}{8} |E_z| - 10\eta^{3/2} |E| - |E_\pi|\\
        &\geq \expect[|\delta(S)|] + \eta^{3/2} |E|.
    \end{align*}
\end{proof}

We have shown that if either $E_z$ or $E_{\pi}$ contain a constant fraction of $E$, then we can improve upon Goemans-Williamson, so we may assume from now on that both of these sets have negligible size.
Now we will utilize the fact that $G$ has an edge-disjoint triangle packing of size $t \cdot |E|$ to show that there are many edges not near the critical angle $\critangle$.
Define the set
\[
    E' = \{e \in E \mid \theta_e \leq \frac{2\pi}{3}\} \setminus E_z.
\]

\begin{lemma} \label{lem:eprimebig}
    $|E'| \geq t \cdot |E| - |E_z|$.
\end{lemma}

\begin{proof}
    Fix any triangle $T = \{uv, vw, wu\}$.
    We can bound the sum of angles $\theta_{uv} + \theta_{vw} + \theta_{uw} \leq 2 \pi$, and so at least one $e \in T$ has $\theta_e \leq \frac{2 \pi}{3}$.
    Thus, in any triangle $T$, we have that $T \cap (E' \cup E_z)$ is non-empty.
    Because $G$ has $t \cdot |E|$ edge-disjoint triangles, there are at least $t \cdot |E|$ edges in $(E' \cup E_z)$.
    Subtracting out those edges in $E_z$ completes the proof.
\end{proof}

We note that $\min_{\theta \in [0, \frac{2\pi}{3}]} \frac{2}{\pi} \frac{\theta}{1 - \cos \theta} = \frac{8}{9} > \alphagw + 0.01$, which motivates the following lemma.

\begin{lemma} \label{lem:betterrounding}
    If $t \geq 97 \sqrt \eta$ and $|E_z| \leq 96 \sqrt \eta |E|$, then 
    \[
        \expect[|\delta(S)|] \geq (\alphagw + 10^{-4} \eta^{3/2}) \cdot \maxcut(G).
    \]
\end{lemma}

\begin{proof}
    The assumptions on $t$ and $|E_z|$ imply, through Lemma~\ref{lem:eprimebig}, that $|E'| \geq \sqrt \eta |E|$.
    We calculate
    \begin{align*}
        \expect[|\delta(S)|] &= \sum_{e \in E'} \frac{\theta_e}{\pi} + \sum_{e \in E \setminus E'} \frac{\theta_e}{\pi} \\
        &\geq \frac{8}{9} \sum_{e \in E'} \frac{1 - \cos \theta_e}{2} + \alphagw \sum_{e \in E \setminus E'} \frac{1 - \cos \theta_e}{2} \\
        &\geq \alphagw \cdot SDP^* + 0.01 \sum_{e \in E'} \frac{1 - \cos \theta_e}{2} \\
        &\geq \alphagw \cdot SDP^* + 10^{-3} \sum_{e \in E'} \frac{\theta_e}{\pi} \\
        &\geq \alphagw \cdot SDP^* + 10^{-4} \cdot \eta \cdot |E'| \\
        &\geq \alphagw \cdot SDP^* + 10^{-4} \cdot \eta^{3/2} \cdot |E| \\
        &\geq (\alphagw + 10^{-4} \cdot \eta^{3/2}) \cdot \maxcut(G).        
    \end{align*}
\end{proof}

Now we can recall and prove Theorem~\ref{thm:apxratio}.
\apxratio*

\begin{proof}

Note that the definition of $\eta$ implies $\eta < 0.01$ and $t > 97 \sqrt \eta$.
If $|E_\pi| \geq \eta^{3/2} |E|$, then apply Lemma~\ref{lem:manylarge} to find that 
\[
    \expect[|\delta(S)|] \geq (\alphagw + 10^{-2} \eta^{3/2}) \cdot \maxcut(G) = (\alphagw + 10^{-8} t^3) \cdot \maxcut(G).
\]
If $|E_\pi| < \eta^{3/2} |E|$ and $|E_z| \geq 96 \sqrt \eta |E|$, then apply Lemma~\ref{lem:manysmall} to find that
\[
    \expect[|\delta(S')|] \geq (\alphagw + \eta^{3/2}) \cdot \maxcut(G) = (\alphagw + 10^{-6} t^3) \cdot \maxcut(G).
\]
Finally, if $|E_\pi| < \eta^{3/2} |E|$ and $|E_z| < 96 \sqrt \eta |E|$, then apply Lemma~\ref{lem:betterrounding} to find that
\[
    \expect[|\delta(S')|] \geq (\alphagw + 10^{-4}\eta^{3/2}) \cdot \maxcut(G) = (\alphagw + 10^{-10} t^3) \cdot \maxcut(G).
\]
\end{proof}

\section{Hardness of Approximating Maximum Cut on Split Graphs} \label{section:hardness}

In this section, we show that, subject to the Small Set Expansion Hypothesis, \mcproblem on split graphs is hard to approximate to some constant.
Our starting point is hardness of finding a small balanced cut on weighted graphs.
Given a graph $G = (V, E)$ with weights $w_e \in [0, 1]$ for $e \in E$, we define
\[
    \mu(S) := \frac{\sum_{u \in S} d_u}{\sum_{v \in V} d_v}
\]
as the \emph{normalized set size} of $S \subseteq V$.
Here, $d_u := \sum_{e \in \delta(u)} w_e$ is defined as the \emph{weighted} degree of $u$.
In an unweighted regular graph, we have that $\mu(S) = \frac{|S|}{|V|}$.
Also define
\[
    \Phi(S) := \frac{\sum_{e \in \delta(S)} w_e}{\sum_{u \in S} d_u}
\]
as the \emph{edge expansion} of $S \subseteq V$.
In an unweighted $d$-regular graph, we have that $\Phi(S) = \frac{|\delta(S)|}{d|S|}$.

\begin{lemma}[Corollary~3.6 of \cite{raghavendra2012reductions}] \label{lem:hard1}
    There is a constant $c_1 > 0$ such that for any sufficiently small $\epsilon > 0$, it is SSE-hard to distinguish between the following two cases for a weighted graph $G = (V, E)$: \\

    \textbf{Yes:} There exists $S \subseteq V$ such that $\mu(S) = \frac{1}{2}$ and $\Phi(S) \leq 2 \epsilon$. \\
    
    \textbf{No:} Every $S \subseteq V$ with $\mu(S) \in [\frac{1}{10}, \frac{1}{2}]$ satisfies $\Phi(S) \geq c_1 \sqrt \epsilon$.
\end{lemma}

To utilize this hardness, we first create an unweighted instance as follows.
\begin{lemma} \label{lem:hard2}
    There is a constant $c_2 > 0$ such that for any sufficiently small $\epsilon, \eta > 0$, it is SSE-hard to distinguish between the following two cases for an unweighted, non-simple graph $G = (V, E)$ with $|E| \leq |V|^5$: \\

    \textbf{Yes:} There exists $S \subseteq V$ such that $\mu(S) \in [\frac{1}{2} - \eta, \frac{1}{2} + \eta]$ and $\Phi(S) \leq 3 \epsilon$. \\
    
    \textbf{No:} Every $S \subseteq V$ with $\mu(S) \in [\frac{1}{10} + \eta, \frac{1}{2}]$ satisfies $\Phi(S) \geq c_2 \sqrt \epsilon$.
\end{lemma}

\begin{proof}
    We begin with a gap instance $G = (V, E)$ of the form in Lemma~\ref{lem:hard1}.
    We will create a weighted instance $G'$ with the same vertex and edge set such that each edge has weight in $\{1, 2, \ldots, n^3\}$.
    Duplicating edges according to their weight will then yield the desired statement for unweighted graphs.

    First, we may assume without loss of generality that the edges in $G$ are scaled such that the sum of degrees is $\sum_{v \in V} d_v = n^3$.
    This is because multiplying the weight of all edges by an equal constant factor does not change the results of $\mu$ or $\Phi$.
    Now produce the weights $\{w_e'\}_{e \in E}$ by setting $w_e' := \lfloor w_e \rfloor$.
    Note in particular that this implies $w_e' \leq |V|^3$ for all $e \in E$.
    Define $\mu', \Phi'$, and $d_v'$ analogously to $\mu, \Phi$, and $d_v$, except for weights $\{w_e'\}_{e \in E}$.
    Note that $w_e \geq w_e' > w_e - 1$ for all $e \in E$ and so $d_v \geq d_v' > d_v - n$ for all $v \in V$.

    Suppose that $G$ is in the \textbf{Yes} case of Lemma~\ref{lem:hard1}, and let $S \subseteq V$ have $\mu(S) = \frac{1}{2}$ and $\Phi(S) \leq 2 \epsilon$.
    Then we calculate
    \begin{align*}
        \mu'(S) &= \frac{\sum_{u \in S} d_u'}{\sum_{v \in V} d_v'} \\
        &\geq \frac{\sum_{u \in S} d_u - n |S|}{\sum_{v \in V} d_v} \\
        &= \mu(S) - \frac{n |S|}{n^3} \\
        &\geq \frac{1}{2} - \frac{1}{n}.
    \end{align*}
    Thus, for sufficiently large $n$, we have that $\mu'(S) \geq \frac{1}{2} - \eta$.
    By a similar calculation, we have that $\mu'(V \setminus S) \geq \mu(V \setminus S) - \eta = \frac{1}{2} - \eta$ and so $\mu'(S) = 1 - \mu'(V \setminus S) \leq \frac{1}{2} + \eta$.
    We now aim to show that $\Phi'(S) \leq 3 \epsilon$.
    We calculate
    \begin{align*}
        \Phi'(S) &= \frac{\sum_{e \in \delta(S)} w_e'}{\sum_{u \in S} d_u'} \\
        &\leq \frac{\sum_{e \in \delta(S)} w_e}{\sum_{u \in S} d_u - n |S|} \\
        &\leq \frac{\sum_{e \in \delta(S)} w_e}{( 1- 2/n) \sum_{u \in S} d_u} \\
        &= \frac{1}{1 - 2/n} \Phi(S) \\
        &\leq \frac{1}{1 - 2/n} 2 \epsilon.
    \end{align*}
    Here, the second inequality uses the fact that $\mu(S) = \frac{1}{2}$ and so $\sum_{u \in S} d_u = \frac{n^3}{2}$.
    Thus, for sufficiently large $n$, we have that $\Phi'(S) \leq 3 \epsilon$.
    This establishes that the \textbf{Yes} case of $G$ maps to the \textbf{Yes} case of $G'$.

    Suppose that $G$ is in the \textbf{No} case of Lemma~\ref{lem:hard1}.
    Then any $S \subseteq V$ with $\mu(S) \leq \frac{1}{2}$ either has $\mu(S) < \frac{1}{10}$ or $\Phi(S) \geq c_1 \sqrt \epsilon$.
    Take any $S \subseteq V$. 
    If $\mu(S) < \frac{1}{10}$, then $\mu'(S) \leq \mu(S) + \eta < \frac{1}{10} + \eta$.
    So consider the case in which $\mu(S) \geq \frac{1}{10}$.
    Then we calculate
    \begin{align*}
        \Phi'(S) &= \frac{\sum_{e \in \delta(S)} w_e'}{\sum_{u \in S} d_u'} \\
        &\geq \frac{\sum_{e \in \delta(S)} w_e - n^2}{\sum_{u \in S} d_u} \\
        &= \Phi(S) - \frac{n^2}{\sum_{u \in S} d_u} \\
        &\geq \Phi(S) - \frac{10}{n} \\
        &\geq c_1 \sqrt \epsilon - \frac{10}{n}.
    \end{align*}
    If we set $c_2 = \frac{c_1}{2}$, then for sufficiently large $n$, $\Phi'(S) \geq c_2 \sqrt \epsilon$.
    This establishes that the \textbf{No} case of $G$ maps to the \textbf{No} case of $G'$.
\end{proof}

Our final reduction to \mcproblem requires our graph to be regular and $|E| \in \mathcal O(n^2)$, so before proceeding, we must first take a standard sparsification step.

\begin{lemma} \label{lem:hard3}
    There is a constant $c_3 > 0$ such that for any sufficiently small $\epsilon, \zeta > 0$, it is SSE-hard under randomized reductions to distinguish between the following two cases for an unweighted, non-simple graph $G = (V, E)$ with $|E| \in \Theta(|V|^2)$ and maximum degree at most $(1 + \zeta) d$, where $d$ is the minimum degree: \\

    \textbf{Yes:} There exists $S \subseteq V$ such that $\mu(S) \in [\frac{1}{2} - \zeta, \frac{1}{2} + \zeta]$ and $\Phi(S) \leq 4 \epsilon$. \\
    
    \textbf{No:} Every $S \subseteq V$ with $\mu(S) \in [\frac{1}{10} + \zeta, \frac{1}{2}]$ satisfies $\Phi(S) \geq c_3 \sqrt \epsilon$.
\end{lemma}

\begin{proof}
    We begin with a gap instance $G = (V, E)$ of the form in Lemma~\ref{lem:hard2}.
    We randomly create $G' = (V', E')$ as follows.
    For each vertex $u \in V$, create $d_u$ copies and add them to $G'$: 
    \[
        V' = \{u_1, u_2, \ldots, u_{d_u} \mid u \in V\}.
    \]
    Let $g : V' \rightarrow V$ map clones $u_i \in V'$ to their original vertex $u \in V$.
    For each edge $uv \in E$ and clones $u' \in g^{-1}(u), v' \in g^{-1}(v)$, let $p_{uv} := \frac{R}{d_u d_v}$ for a parameter $R$ we will adjust later.
    Now add $\lfloor p_{uv} \rfloor$ copies of $u' v'$ to $E'$ and randomly add a final edge with probability $p_{uv} - \lfloor p_{uv} \rfloor$.
    Note that the expected number of $u'v'$ edges is equal to $p_{uv}$.

    Let us consider the properties of $G'$. We have that $|V'| = \sum_{v \in V} d_v = 2 |E|$ and
    \[
        \expect[|E'|] = \sum_{uv \in E} p_{uv} \cdot d_u \cdot d_v = R |E|.
    \]
    Let $q_e$ be the random variable denoting the number of edges in $G'$ ``produced'' from $e \in E$.
    By a standard application of Chernoff bounds and the union bound over $2^{V'}$ and $E$, we can show, for any constant $\alpha > 0$, that
    \begin{enumerate}
        \item $(1 - \alpha) \expect[|\delta_{G'}(S)|] \leq |\delta_{G'}(S)| \leq (1 + \alpha) \expect[|\delta_{G'}(S)|]$ for all $S \subseteq V'$ and
        \item $(1 - \alpha) R \leq q_e \leq (1 + \alpha) R$ for all $e \in E$
    \end{enumerate}
    with high probability, assuming $R \in \Omega(n)$.
    In particular, item 1 implies that $d_v' \in [(1 - \alpha) R, (1 + \alpha) R]$ for all $v \in V'$, where $d_v'$ is the degree of $v$ in $G'$.
    Suppose from now on that items 1 and 2 are both true.
    Now suppose $G$ is a \textbf{Yes} instance of Lemma~\ref{lem:hard2}.
    That is, there is a $S \subseteq V$ such that $\mu(S) \in [\frac{1}{2} - \eta, \frac{1}{2} + \eta]$ and $\Phi(S) \leq 3 \epsilon$.
    Let $S' := g^{-1}(S)$, and $\mu'$ be defined for $G'$ as $\mu$ is defined for $G$.
    Then we calculate
    \begin{align*}
        \mu'(S') &= \frac{\sum_{u \in S} d_u'}{\sum_{v \in V} d_v'} \\
        &\geq \frac{(R- \alpha) |S'|}{(R+ \alpha) |V'|} \\
        &= \frac{R - \alpha}{R + \alpha} \frac{\sum_{u \in S} d_u}{\sum_{v \in V} d_v} \\
        &= \frac{R - \alpha}{R + \alpha} \mu(S) \\
        &\geq \frac{R - \alpha}{R + \alpha} (\frac{1}{2} - \eta).
    \end{align*}
    Applying the same calculation for $\mu'(V' \setminus S')$ yields the upper bound
    \[
        \mu'(S') = 1 - \mu'(V \setminus S') \leq 1 - \frac{R - \alpha}{R + \alpha} (\frac{1}{2} - \eta).
    \]
    For sufficiently small $\alpha > 0$, this implies that $\mu'(S') \in [\frac{1}{2} - \zeta, \frac{1}{2} + \zeta]$.
    Let $\Phi'$ be defined for $G'$ as $\Phi$ for $G$.
    We calculate
    \begin{align*}
        \Phi'(S') &= \frac{|\delta_{G'}(S')|}{\sum_{u \in S'} d_u} \\
        &\leq \frac{(R + \alpha) |\delta_G(S)|}{(R - \alpha) \sum_{u \in S} d_u} \\
        &= \frac{R+\alpha}{R - \alpha} \Phi(S) \\
        &\leq \frac{R + \alpha}{R - \alpha} 3 \epsilon.
    \end{align*}
    For sufficiently small $\alpha > 0$, we get the upper bound $\Phi'(S') \leq 4 \epsilon$.
    This establishes that the \textbf{Yes} case of $G$ maps to the \textbf{Yes} case of $G'$ with high probability.
    
    Now suppose that $G$ is a \textbf{No} instance of Lemma~\ref{lem:hard2}.
    Consider any $S' \subseteq V'$ with $\mu'(S') \in [\frac{1}{10} + \zeta, \frac{1}{2}]$.
    Let $f \in [0, 1]^V$ be a vector defined by $f_v = \frac{|S' \cap g^{-1}(v)|}{d_v}$.
    That is, $f$ indicates the proportion of each original vertex selected by $S'$.
    We calculate
    \begin{align*}
        \Phi'(S') &= \frac{|\delta_{G'}(S')|}{\sum_{v \in V'} d_v'} \\
        &\geq \frac{(1 - \alpha)\expect[|\delta_{G'}(S')|]}{(R + \alpha) |V'|} \\
        &= \frac{1 - \alpha}{R + \alpha}\frac{\sum_{uv \in E} p_{uv} \cdot ((1 - f_u) d_u f_v d_v + f_u d_u (1 - f_v) d_v)}{|V'|} \\
        &= \frac{R (1 - \alpha)}{R + \alpha} \frac{\sum_{uv \in E} (1 - f_u) f_v + f_v (1 - f_u)}{|V'|} \\
    \end{align*}
    Let $S \subseteq V$ be a random variable sampled by including $v \in S$ with probability $f_v$.
    Then we have that
    \begin{align*}
        \expect[\Phi(S)] &= \frac{\expect[|\delta_G(S)|]}{\sum_{v \in V} d_v} \\
        &= \frac{\sum_{uv \in E} (1 - f_u) f_v + f_u (1 - f_v)}{|V'|} \\
        &\leq \frac{R + \alpha}{R (1 - \alpha)} \Phi'(S').
    \end{align*}
    We can apply Chernoff bounds on $|\delta_G(S)|$ to find that $\Phi(S) < (1 + \alpha) \expect[\Phi(S)]$ with high probability.
    We additionally calculate
    \begin{align*}
        \expect[\mu(S)] &= \frac{\sum_{u \in V} f_u d_u}{\sum_{v \in V} d_v} \\
        &= \frac{|S'|}{|V'|} \\
        &\geq \frac{R - \alpha}{R + \alpha} \mu(S) \\
        &\geq \frac{R - \alpha}{R + \alpha} (\frac{1}{10} + \zeta).
    \end{align*}
    By setting $\alpha, \eta > 0$ sufficiently small and applying Chernoff bounds, this implies that $\expect[\mu(S)] > \frac{1}{10} + \eta$ with probability at least $\frac{1}{2}$.
    Thus, by the probabilistic method, there exists some $S \subseteq V$ such that $\mu(S) > \frac{1}{10} + \eta$ and $\Phi(S) < (1 + \alpha) \expect[\Phi(S)]$.
    Applying Lemma~\ref{lem:hard2}, we find that
    \begin{align*}
        \Phi'(S') &\geq \frac{R(1 - \alpha)}{R + \alpha} \expect[\Phi(S)] \\
        &> \frac{R (1 - \alpha)}{(R + \alpha) (1 + \alpha)} \Phi(S) \\
        &\geq \frac{R (1 - \alpha)}{(R + \alpha) (1 + \alpha)} c_2 \sqrt \epsilon.
    \end{align*}
    
    So, set $c_3 := \frac{R (1 - \alpha)}{(R + \alpha) (1 + \alpha)} c_2$.
    This establishes that the \textbf{No} case of $G$ maps to the \textbf{No} case of $G'$ with high probability.
\end{proof}

We now convert the language of $\mu$ and $\Phi$ to the simpler language of cardinalities of sets.
\begin{lemma} \label{lem:ssegap}
    There exists a constant $c' > 0$ such that for all sufficiently small $\varepsilon, \eta > 0$, it is SSE-hard under randomized reductions to distinguish between the following cases for an unweighted graph $G = (V, E)$ with $|E| \in \Theta(|V|^2)$: \\

    \textbf{Yes:} There exists a cut $S \subseteq V$ with $(\frac{1}{2} - \eta) |V| \leq |S| \leq \frac{|V|}{2}$ such that $|\delta(S)| \leq 10 \varepsilon |E|$. \\

    \textbf{No:} For all cuts $S \subseteq V$ with $|S| \leq \frac{|V|}{2}$, either $|S| \leq \frac{|V|}{5}$ or $|\delta(S)| \geq c' \sqrt \varepsilon |E|$.
\end{lemma}

\begin{proof}
    We begin with a gap instance $G = (V, E)$ of the form in Lemma~\ref{lem:hard3}.
    We will not need to modify this instance to produce our desired result.
    Suppose that $G$ is a \textbf{Yes} instance of Lemma~\ref{lem:hard3}.
    Then there is some $S \subseteq V$ with $\mu(S) \in [\frac{1}{2} - \zeta, \frac{1}{2} + \zeta]$ and $\Phi(S) \leq 4 \epsilon$.
    Using that every vertex $v \in V$ has degree $d_v \leq (1 + \zeta) d$, where $d$ is the minimum degree of $G$, we can bound
    \begin{align*}
        \mu(S) = \frac{\sum_{u \in S} d_u}{\sum_{v \in V} d_v} 
        \leq \frac{|S| (1 + \zeta) d}{|V| d}.
    \end{align*}
    Thus, $|S| \geq \frac{1/2 - \zeta}{1 + \zeta} |V|$.
    For sufficiently small $\zeta > 0$, we can then bound below $|S| \geq (\frac{1}{2} - \eta) |V|$.
    Because $\mu(S) = \mu(V \setminus S)$, we also bound $|V \setminus S| \geq (\frac{1}{2} - \eta) |V|$.
    If $|S| > \frac{|V|}{2}$, then swap $S$ with $V \setminus S$, noting that $\delta(S) = \delta(V \setminus S)$.
    This fulfills the $(\frac{1}{2} - \eta) |V| \leq |S| \leq \frac{|V|}{2}$ condition.
    Similarly, we can bound
    \[
        \Phi(S) = \frac{|\delta(S)|}{\sum_{u \in S} d_u} \geq \frac{|\delta(S)|}{|V| d (1 + \epsilon)}.
    \]
    With the bound of $\Phi(S) \leq 4 \epsilon$ from Lemma~\ref{lem:hard3}, this implies that
    \[
        |\delta(S)| \leq 4 \epsilon |V| d (1 + \zeta) \leq 10 \epsilon |E|
    \]
    for sufficiently small $\zeta > 0$. Thus, $G$ being a \textbf{Yes} instance of Lemma~\ref{lem:hard3} implies the \textbf{Yes} conditions of this lemma.

    Suppose that $G$ is a \textbf{No} instance of Lemma~\ref{lem:hard3}. Consider any $S \subseteq V$. Suppose first that $\mu(S) < \frac{1}{10} + \zeta$.
    Then we have that
    \[
        \mu(S) = \frac{\sum_{u \in S} d_u}{\sum_{v \in V} d_v} \geq \frac{|S| d}{|V| (1 + \zeta) d},
    \]
    and so $|S| \leq (\frac{1}{10} + \zeta) (1 + \zeta) |V|$.
    For $\zeta$ sufficiently small, this implies $|S| \leq \frac{|V|}{5}$.
    Suppose otherwise, that $\mu(S) \geq \frac{1}{10} + \zeta$.
    Then we have that
    \[
        \Phi(S) = \frac{|\delta(S)|}{\sum_{u \in S} d_u} \leq \frac{|\delta(S)|}{|V| d},
    \]
    and so $|\delta(S)| \geq 3 d \sqrt \epsilon |V| \geq \frac{6}{1 + \zeta} \sqrt \epsilon |E|$.
    Setting $c' := \frac{6}{1 + \zeta}$ finishes the proof.
\end{proof}

We now reduce from gap instances of the type in Lemma~\ref{lem:ssegap} to show hardness for \mcproblem on split graphs.

\splithard*

\begin{proof}
    We reduce from a graph $G = (V, E)$ of the form in Lemma~\ref{lem:ssegap} to a (simple) split graph $G' = (V' = K \cup I, E')$ as follows.
    Let the clique portion of $G'$ be $K := V$, and let the independent portion of $G'$ be $I := \{v_e\}$. 
    We define the edge set of $G'$ as $E' := \{u v_e \mid u \in V, e \in u\} \cup \{uw \mid u, w \in V\}$.
    That is, we place a copy of $V$ in the clique portion of $G'$, and place one vertex for each edge of $E$ in the independent set portion, each connected to its two endpoints in $V$.

    Note that in a split graph, the maximum cut is defined solely by its intersection with the clique portion $K$ of the graph, as the decision for vertices in the independent set portion $I$ can be made greedily.
    Define $\delta'(S)$ for $S \subseteq K = V$ to be the edges in the maximum cut of $G'$ defined by $S$, breaking ties arbitrarily. 

    Suppose that $G$ is a \textbf{Yes} instance as defined in Lemma~\ref{lem:ssegap}.
    Then by direct counting, we have that, in $G'$,
    \[
        \maxcut(G') \geq |\delta'(S)| \geq (0.25 - \eta^2) |V|^2 + (2 - 10\varepsilon) |E|.
    \]
    Define $\omega := (0.25 - \eta^2) |V|^2 + (2 - 10\varepsilon) |E|$.
    Suppose instead that $G$ is a \textbf{No} instance and consider any $S \subseteq V$ with $|S| \leq 0.5 |V|$.
    We will show that, with the right choice of $\varepsilon$ and $\eta$,  $|\delta'(S)| \leq (1 - c) \omega$ for some constant $c > 0$.
    If $|S| \leq  \frac{|V|}{5}$, then 
    \[
        |\delta'(S)| \leq \frac{|V|^2}{5} + 2 |E|.
    \]
    Because $|E| \in \Theta(|V|^2)$, and $\frac{1}{5} < 0.25 - \eta^2$ for sufficiently small $\eta$, this implies that $|\delta'(S)| < (1 - c) \omega$ for some constant $c > 0$ when $\eta$ and $\varepsilon$ are sufficiently small.
    Otherwise, if $|S| \geq \frac{|V|}{5}$, then we must have that $|\delta_G(S)| \geq c' \sqrt \varepsilon |E|$. 
    Then
    \[
        |\delta'(S)| \leq 0.25|V|^2 + (2 - c' \sqrt \varepsilon) |E|.
    \]
    As above, because $2 - c' \sqrt \varepsilon < 2 - 10 \varepsilon$ for sufficiently small $\varepsilon$, we have that $|\delta'(S)| < (1 - c) \omega$ for some constant $c > 0$ when $\eta$ and $\varepsilon$ are sufficiently small.
    Thus, in the \textbf{No} case, we have that $\maxcut(G') \leq (1 - c) \omega$.
    Therefore, it is SSE-hard to distinguish between $\maxcut(G') \geq \omega$ and $\maxcut(G') \leq (1 - c) \omega$.
    This implies SSE-hardness of approximating \mcproblem to a factor of $1-c$ on split graphs.
\end{proof}

It is critical that Lemma~\ref{lem:ssegap} allows for non-simple graphs.
If $G$ were simple, then our reduction in the proof of Theorem~\ref{thm:splithard} results in the relation $\maxcut(G') = \maxcut(G^c) + 2 |E|$, where $G^c = (V, E^c)$ is the complement of $G$.
This is proved explicitly in \cite{bodlaender2000complexity}.
Then, if $|E| \in \omega(|E^c|)$, the value $2 |E|$ is a good estimate for $\maxcut(G')$. 
If $|E| \in \mathcal O(|E^c|)$, then $|E^c| \in \Omega(|V|^2)$.
In this case, $G^c$ is a dense graph and so there is a PTAS for $\maxcut(G^c)$ \cite{fernandez1996max,mathieu2008yet}.
In particular, these cases imply that there is a PTAS for $\maxcut(G')$ whenever $G'$ is a 2-split graph (i.e., all vertices in $I$ have degree $2$) without any ``duplicated'' vertices of $I$ that have the same neighborhood.

\bibliographystyle{alpha}

\bibliography{intervalmaxcut}

\end{document}